\documentclass[11pt]{article}

\usepackage{amsmath, amsthm, amsfonts}
\usepackage[margin=2.5cm]{geometry}
\usepackage{amssymb}
\usepackage{fancyhdr}
\usepackage[cp1250]{inputenc}

\newtheorem{thm}{Theorem}[section]

\usepackage{hyperref}
\hypersetup{
    colorlinks=true,
    linkcolor=blue,
    filecolor=magenta,      
    urlcolor=cyan,
}

\title{Minimal Number of Observables for Quantum Tomography of Systems with Evolution Given by Double Commutators}
\author{Artur Czerwi{\'n}ski\\
\\ \small ResearchGate: \url{www.researchgate.net/profile/Artur_Czerwinski} \\
\small E-mail: aczerwin@fizyka.umk.pl \\
  \small 1. Institute of Physics\\
  \small Nicolaus Copernicus University\\
  \small 87-100 Toru{\'n}\\
	\small 2. Center for Theoretical Physics \\
\small Polish Academy of Sciences \\
\small 02-668 Warszawa\\
}

\begin{document}
\maketitle

\begin{abstract}
In this paper we analyze selected evolution models of  $N-$level open quantum systems in order to find the minimal number of observables (Hermitian operators) such that their expectation values at some time instants determine the accurate representation of the quantum system. The assumption that lies at the foundation of this approach to quantum tomography claims that time evolution of an open quantum system can be expressed by the Kossakowski - Lindblad equation of the form $\dot{\rho} = \mathbb{L} \rho$, which is the most general type of Markovian and time-homogeneous master equation which preserves trace and positivity. We consider the cases when the generator of evolution can be presented by means of two or more double commutators. Determining the minimal number of observables required for quantum tomography can be the first step towards optimal tomography models for $N-$level quantum systems. 
\end{abstract}


\section{Introduction}

In this paper we shall denote the Hilbert space by $\mathcal{H}$ and we shall assume that $dim \mathcal{H} = N < \infty $. $B(\mathcal{H})$ shall refer the complex vector space of all bounded linear operators in $\mathcal{H}$. Naturally, $B(\mathcal{H})$ is isomorphic with the space of N-dimensional complex matrices. The latter shall be represented by $\mathbb{M}_N (\mathbb{C})$. Finally, $B_* (\mathcal{H})$ shall refer to the real Banach space of self-adjoint (hermitian) operators on $\mathcal{H}$. In Physics the elements of $B_* (\mathcal{H})$ are called observables due to the fact that an element from this space is assigned to every measurable quantity .

The term \textit{quantum tomography} is used in reference to all methods and approaches which aim to reconstruct the accurate representation of a quantum system by conducting a series a measurements. Among a large number of approaches to quantum tomography one can especially mention the so-called static model of tomography, which requires $N^2 - 1$ measurements of distinct observables taken at time instant $t=0$ (see more in \cite{altepeter04,alicki87,genki03}). Evidently, as the number of observables increases with the square of $N$ this approach seems rather impracticable. Therefore, there is a need to devise other, more efficient, models - for example in \cite{ole95} one can read about reconstruction of the density matrix by simple projectors. A paper published in 2011 introduced a new approach to quantum tomography which is based upon weak measurement. In that paper it has been proved that the wave function of a pure state can be measured directly \cite{bamber}, in a contrast to a common belief. Subsequent papers revealed that this approach can be generalized also for mixed state identification \cite{wu}. 

In this paper we follow yet another approach to quantum tomography - the stroboscopic tomography (or stroboscopic observability) which initiated in 1983 in the article \cite{jam83}. This approach was developed in further papers, such as \cite{jam00,jam04}. One can also refer to a well-written review article \cite{jam12}, which contains all fundamental results. The main advantage of this method is that it enables to determine the optimal criteria for observabiliy of an open quantum system. Thus it seems to have considerable potential applications in experiments.

The assumption that lies at the very foundation of the stroboscopic tomography claims that the evolution of an open quantum system can be expressed by a master equation of the form
\begin{equation}\label{eq:kossak}
\dot{\rho} = \mathbb{L} \rho,
\end{equation}
where the operator $\mathbb{L}$ shall be referred to as the generator of evolution and its most general structure in diagonalized form can be expressed as \cite{gorini76,lindblad76}
\begin{equation}\label{eq:general}
\mathbb{L} \rho = - i [H,\rho] + \frac{1}{2} \sum_{i=1} ^{N^2 -1} \gamma_ i \left ( [V_i \rho, V_i ^*] + [V_i, \rho V_i ^*]  \right ),
\end{equation}
where $H\in B_* (\mathcal{H})$ and $\gamma_i \geq 0 $. Operators $V_i \in B(\mathcal{H})$ are called Lindblad operators.

One can easily notice that if $V_i = V_i ^*$ for $i=1,\dots,N^2 -1$, one is able to rewrite the general form of the generator of evolution as
\begin{equation}\label{eq:commutators}
\mathbb{L} \rho = - i [H,\rho] - \frac{1}{2} \sum_{i=1} ^{N^2 -1} \gamma_ i [V_i, [V_i, \rho]],
\end{equation}
where the dissipative part of the generator is given by means of double commutators. Generators of the form given by \eqref{eq:commutators} have been the subject of many analysis in quantum Physics. Recent results introduced in \cite{kasia15} show that the double commutator form of the generator is convenient for analyzing qubit decoherence in the framework of geometric quantum mechanics. Thus the main motivation for the current paper comes from the observation that double commutators appear in the structure of the generator of evolution under the assumption that the Lindblad operators are Hermitian along with the fact that the generator of structure given by \eqref{eq:commutators} is considered in other aspects of quantum Physics.

In order to determine the initial density matrix $\rho(0)$ (and on the basis of the knowledge about the evolution - the complete trajectory of the state) one needs to assume the availability of a set of identically prepared quantum systems, each with evolution given by the master equation with the generator $\mathbb{L}$. Furthermore, one has to bear in mind that each system can be measured only once, because any measurement, in general, changes the state. 

In the stroboscopic tomography one assumes to have a set of $r$ distinct hermitian operators denoted by $\{Q_1, \dots, Q_r \}$ ($Q_i^* = Q_i$), where $r<N^2 -1$. Each operator can be measured at one or more time instants from the set $\{t_1, \dots, t_p\}$. The measurement results, denoted by $m_i (t_j)$, are represented by the formula $m_i (t_j) = Tr(Q_i \rho(t_j))$. From many possible research problems connected with the stroboscopic tomography, in this paper we shall discuss the aspect of the minimal number of observables required for quantum tomography, which is considered the most important.  One can recall the theorem \cite{jam04}.

\begin{thm}\label{thm:1}
For a quantum system with the time evolution given by the Kossakowski-Lindblad master equation of the form \eqref{eq:kossak}, there exists a number (denoted by $\eta$) which determines the minimal number of observables required to reconstruct the density matrix. It can be computed from the equality
\begin{equation}
\eta := \max \limits_{\lambda \in \sigma (\mathbb{L})} \{ dim Ker (\mathbb{L} - \lambda \mathbb{I})\}.
\end{equation}
The number $\eta$ shall be called \textit{the index of cyclicity} of the quantum system.
\end{thm}

If one wants to employ a quantum tomography model in an experiment, then the index of cyclicity $\eta$ tells them how many distinct experimental setups one needs to prepare to be able to reconstruct the density matrix. In general, it is more efficient to repeat the same kind of measurement a few times then to prepare a few different kinds of measurements. Therefore, we believe that the stroboscopic tomography possesses significant potential for future applications as it focuses on determining the optimal criteria for quantum tomography, i.e. the minimal number of distinct observables. One can observe that the value of $\eta$ depends on the geometric properties of the generator of evolution, but the meaning of $\eta$ is purely physical.

The well-known results concerning the index of cyclicity for $N-$level quantum systems relate to the time evolution given by a von Neumann equation (see the result in \cite{jam83}) and Gaussian semigroups, i.e. the generator of evolution given by a double commutator $\mathbb{L} \rho = - \frac{1}{2} [H, [H,\rho]]$ where $H \in B_* (\mathcal{H})$ (see more in \cite{jam04}). In this article we consider more general evolution models when the generator is given by two or more double commutators. From mathematics we can adopt methods that allow to analyze such operators with respect to the eigenvalues and minimal polynomial (see for example \cite{marcus71}).

This paper shows that under certain assumptions one can successfully calculate the index of cyclicity for an $N-$level open quantum system for a variety of evolution models. This work can be the first step towards complete and optimal tomography models for $N-$level quantum systems.

\section{Generator of evolution with two Lindblad operators}

In this section we shall analyze an open quantum system with the generator of evolution that consists of two Lindblad operators, which shall be denoted by $G$ and $F$ (of course, $F,G \in B(\mathcal{H})$). Thus the evolution of the quantum system in question is given by
\begin{equation}\label{eq:1}
\dot{\rho} = \frac{1}{2} \left ( [F \rho, F^*]+[F, \rho F^*] + [G \rho, G^*] + [G, \rho G^*] \right ).
\end{equation}

In order to be able to discuss the problem of observability of a system with such evolution we have to assume a kind of relation between the Lindblad operators $F$ and $G$. In this part we propose to consider the case when $F = G^*$. With this assumption it is possible to rewrite equation \eqref{eq:1} in the form
\begin{equation}\label{eq:2}
\dot{\rho} = \frac{1}{2} \left ( [F \rho, F^*]+[F, \rho F^*] + [F^* \rho, F] + [F^*, \rho F] \right ).
\end{equation}

One can easily observe that the generator of evolution can be expressed by means of two double commutators
\begin{equation}\label{eq:3}
\mathbb{L} \rho = - \frac{1}{2} \left ( [F^*, [F,\rho]] + [F, [F^*, \rho]] \right ). 
\end{equation}

If one wants to obtain a specific result concerning the index of cyclicity, it is necessary to make more assumptions about the operator $F$ and its eigenvalues. First, let us denote by $\sigma (F)$ the spectrum of the operator $F$. If we assume that this operator has $r$ distinct eigenvalues, then the spectrum can be written as
\begin{equation}\label{eq:4}
\sigma (F) = \{ \alpha_1, \dots , \alpha_r \},
\end{equation}
where in general $\alpha_i \in \mathbb{C}$. With each eigenvalue $\alpha_i$ we associate its multiplicity denoted by $n_i$. Obviously, $\sum_{i=1}^r n_i = N$, where $N$ refers to the dimension of the Hilbert space related to the analyzed system.

The generator of evolution in \eqref{eq:3} is still too general to consider the problem of observability. There might be many approaches to specify this generator and make the problem solvable. In this section, in order to get a concrete result, we assume that the operator $F$ is unitary, i.e. $F^* = F^{-1}$. 

Then the explicit form of the generator of evolution $\mathbb{L}$ can be obtained by using the idea of vectorization \cite{henderson81}. One can get
\begin{equation}\label{eq:5}
\mathbb{L} =  (F^{-1})^T \otimes F + F^T \otimes F^{-1}   - 2 \mathbb{I}_N \otimes \mathbb{I}_N.
\end{equation}

For the generator of evolution presented in\eqref{eq:5} we can formulate the following theorem and prove it.

\begin{thm}
The index of cyclicity for the generator of evolution given by \eqref{eq:5} is expressed by 
\begin{equation}\label{thm1}
\eta = max \{ \sum_{i=1}^r  n_i ^2 , \delta_1, \dots, \delta_p \},
\end{equation}
where $\delta_k$ is defined as follows 
\begin{equation}\label{thm2}
\delta_k = 2 \sum_{i=1}^{r-k} n_i n_{i+k}
\end{equation}
and $p = \frac{r-1}{2}$ if $r$ is odd or $ p = \frac{r-2}{2} $ of r is even.
\end{thm}

\begin{proof}

One can instantly notice that  the operator $\mathbb{L}$ is hermitian (self-adjoint), which implies that the algebraic multiplicities of its eigenvalues are the same as the corresponding geometric multiplicities. Therefore, we can focus only on determining the algebraic multiplicities. In order to do that one can realize that \textit{transposition} does not change the spectrum of an operator. This observation leads to the conclusion that the eigenvalues of $\mathbb{L}$ are exactly the same as the eigenvalues of the operator $\mathbb{L}'$ which has the form
\begin{equation}\label{eq:6}
\mathbb{L}' =  F \otimes F ^{-1}+ (F \otimes F^{-1})^{-1}  - 2 \mathbb{I}_N \otimes \mathbb{I}_N.
\end{equation}

It can be noticed that the spectrum of $F \otimes F^{-1}$ is given by
\begin{equation}\label{eq:7}
\sigma (F \otimes F^{-1}) = \{\mu_{ij} \in \mathbb{C}; \mu_{ij} = \frac{\alpha_i}{ \alpha_j} ; \text{  }i,j = 1,\dots, r\}
\end{equation}
To each eigenvalue $\mu_{ij}$ corresponds its multiplicity $ n_i n_j.$ 

Bearing in mind equation \eqref{eq:6}, the spectrum of the generator $\mathbb{L}$ can be written as
\begin{equation}\label{eq:8}
\sigma (\mathbb{L}) = \{\lambda_{ij} \in \mathbb{R}; \lambda_{ij} =  \frac{\alpha_j}{ \alpha_i} + \frac{\alpha_i}{ \alpha_j} - 2; \text{  }i,j = 1,\dots, r\}
\end{equation}
And multiplicity of $\lambda_{ij}$ is equal $n_i n_j $.Based on the equation \eqref{eq:8} it is not possible to uniquely express the multiplicities of $\lambda_{ij}$ by $n_i$. But a specific case can be discussed when the eigenvalues $\alpha_1, \dots, \alpha_r$ are successive terms of a geometric sequence. For clarity it can also be assumed that $|\alpha_1| > |\alpha_2| > \dots > |\alpha_r|$. Let us denote
\begin{equation}\label{eq:9}
\frac{\alpha_{i+1}}{\alpha_i} = q,
\end{equation}
where i = 2,\dots,r-1. 

It is convenient to present the eigenvalues of $\mathbb{L}$ in a form of a matrix $[\lambda_{ij}]$:
\begin{equation}\label{eq:10}
\left[
\begin{matrix}
0 & q+\frac{1}{q} -2 & q^2+\frac{1}{q^2}-2 & q^3+\frac{1}{q^3}-2 &\cdots & q^{r-1}+\frac{1}{q^{r-1}}-2\\
 q+\frac{1}{q} -2 & 0 &  q+\frac{1}{q} -2 &  q^2+\frac{1}{q^2} -2 & \cdots & q^{r-2}+\frac{1}{q^{r-2}} -2 \\
 q^2+\frac{1}{q^2}-2 & q+\frac{1}{q} -2 & 0 &  q+\frac{1}{q}-2 & \cdots &  q^{r-3}+\frac{1}{q^{r-3}}-2 \\
 q^3+\frac{1}{q^3}-2 &  q^2+\frac{1}{q^2}-2 &  q+\frac{1}{q}-2 & 0 & \cdots & q^{r-4}+\frac{1}{q^{r-4}}-2 \\
\vdots & \vdots &  \vdots & \vdots & \ddots & \vdots \\
  q^{r-1}+\frac{1}{q^{r-1}}-2 &  q^{r-2}+\frac{1}{q^{r-2}}-2 & q^{r-3}+\frac{1}{q^{r-3}}-2 &  q^{r-4}+\frac{1}{q^{r-4}} -2 & \cdots & 0
\end{matrix}
\right ].
\end{equation}

It can easily be observed that the above matrix is symmetric and the eigenvalues that lie on the same diagonal are equal. Therefore, let us introduce denotation $\delta_k$ which will refer to the total multiplicity of all eigenvalues that lie on both $k^{th}$ diagonals (where $k=0$  for the main diagonal). According to this definition the following equality holds
\begin{equation}\label{eq:11}
\delta_k = 2 \sum_{i=1}^{r-k} n_i n_{i+k}
\end{equation}
It is sensible to consider only the cases for $ k \in \{1,..., p\}$ where $ p = \frac{r-1}{2}$ if $r$ is odd or $ p = \frac{r-2}{2}$ if $r$ is even, because for $ k > p  $ the obvious inequality holds:
\begin{equation}\label{eq:12}
\sum_{i=1}^r n_i ^2 > \delta_k.
\end{equation}
Therefore, the index of cyclicity for the analyzed generator can be expressed as
\begin{equation}\label{eq:13}
\eta = max\{\sum_{i=1}^r n_i ^2, \delta_1, \dots, \delta_p\},
\end{equation}
which ends the proof.
\end{proof}
In this section we have showed that under specific assumptions it is possible to calculate in a simple way the index of cyclicity for generators of evolution that consist of two Lindblad operators. Interestingly, the obtained final result is exactly the same as the index of cyclicity for the generator given by $\mathbb{L} \rho = - \frac{1}{2} [H,[H,\rho]]$ \cite{jam04}.

\section{Superposition of double commutators as the generator of evolution}

Let us first formulate the problem which will be analyzed in this section. Now we assume to have an operator $F \in B(\mathcal{H})$ which is hermitian, i.e. $F^* = F$, and, as it was mentioned before, $dim \mathcal{H} = N$. Moreover we assume to have a set of matrix polynomials $\{ f_1 , \dots, f_N \}$ ( $f_i: M_N (\mathbb{C}) \rightarrow M_N(\mathbb{C})$) and the degree of each matrix polynomial is not greater that $N$.

Then we shall define opeators $F_k$ in the way
\begin{equation}
F_k := f_k (F) \text{  for  } k=1, \dots, N.
\end{equation}

Now we shall proceed to analyzing another model of evolution of open quantum systems. We propose to consider observability criteria for a generator of evolution given by 
\begin{equation}\label{eq:evolution1}
\mathbb{L} \rho = \sum_{k=1}^N  \gamma_k \left ( [F_k \rho, F_k^*] + [F_k, \rho F_k^*] \right),
\end{equation}
where $\gamma_k \geq 0$ for $k=1,\dots, N$.

One can notice that we have $F_k ^* = F_k$ because $F^* = F$. Thus the evolution equation introduced in \eqref{eq:evolution1} can be expressed by means of double commutators
\begin{equation}\label{eq:evolution}
\mathbb{L} \rho = - \sum_{k=1}^N \gamma_k [F_k , [F_k, \rho ]].
\end{equation}

If one wanted to consider this problem in general, one would observe that it is impossible to reach a specific results due to a large number of parameters (one should bear in mind that each matrix polynomial $f_k$ can in general depend on $N+1$ coefficients). Therefore, in order to be able to draw concrete conclusions concerning conditions for quantum tomography, we propose to analyze the case when the matrix polynomials $\{f_1, \dots, f_N\}$ are defined in the following way
\begin{equation}
f_k (F) := F^k, \text{ for } k=1,\dots,N.
\end{equation}

This assumption allows to rewrite the evolution equation in the form
\begin{equation}\label{eq:model2}
\mathbb{L} \rho = - \sum_{k=1}^N \gamma_k [F^k , [F^k, \rho ]].
\end{equation}

Before introducing the main theorem we shall assume that the spectrum of the operator $F$ has the following structure: $r$ different eigenvalues denoted by $\{ \alpha_1, \cdots, \alpha_r\}$ ($\alpha_i \in \mathbb{R}$) and with each $\alpha_i$ we associate its multiplicity denoted by $n_i$. As $F$ in this model is assumed to be Hermitian, the geometric and algebraic multiplicities are equal. Therefore, there is no need to differentiate between them.

For the generator of evolution given by \eqref{eq:model2} we shall prove the following theorem concerning the index of cyclicity.

\begin{thm}
The index of cyclicity for a quantum system with the generator of evolution given by \eqref{eq:model2} can be expressed by one of the three formulas below.

1. If $\exists_{d\in \mathbb{N}}$ such that $\gamma_{2d-1} \neq 0$, then
\begin{equation}\label{eq:thmeq1}
\eta = \sum_{i=1}^r n_i ^2.
\end{equation}

If $\gamma_{2d-1}=0$ for all $d$ such that $d\in \mathbb{N}$ and $(2d-1)\leq N$ then we may have two situations:

2. if $r = 2u + 1$ for $u \in \mathbb{N}$ then
\begin{equation}\label{eq:thm22}
\eta = \sum_{i=1}^{2u+1} n_i ^2 + 2 \sum_{i=1}^u n_i n_{2u+2-i},
\end{equation}
3. if $r= 2u$ for $ u \in \mathbb{N}$ then
\begin{equation}\label{eq:thm23}
\eta = \sum_{i=1}^{2u} n_i ^2 + 2 \sum_{i=1}^u n_i n_{2u+1-i}.
\end{equation}
\end{thm}
\begin{proof}

To start the proof one needs to notice that the explicit matrix form of the operator $\mathbb{L}$ is given by
\begin{equation}\label{eq:explicit}
\mathbb{L} = - \sum_{k=1}^N \gamma_k \left ( F^{2k} \otimes \mathbb{I}_N - 2 F^k \otimes F^k + \mathbb{I}_N \otimes F^{2k} \right ),
\end{equation}
where the transposition that should appear according to vectorization equality has been skipped due to the fact that it does not influence the spectrum of $\mathbb{L}$.

One can then observe that the spectrum of the operator $\mathbb{L}$ is given by
\begin{equation}\label{eq:spectrum}
\sigma (\mathbb{L}) = \{ \lambda_{ij} \in \mathbb{R}; \text{ } \lambda_{ij} = - \sum_{k=1}^N \gamma_k ( \alpha_i ^k  - \alpha_j ^k )^2, \text{ for } i,j=1,\dots, r\}.
\end{equation}

Important observation in this analysis is that $\lambda_{ij} = 0$ for $\alpha_i = \alpha_j$. Therefore the total multiplicity of $0$ as an eigenvalue of $\mathbb{L}$ (denoted by $m_0$) is at least equal to
\begin{equation}\label{eq:atleast}
m_0 = \sum_{i=1}^r n_i ^2.
\end{equation}

The question is whether this multiplicity can be greater under some additional assumptions concerning the structure of the spectrum of $F$.

To prove point 1. of the theorem one needs to take into account the condition given in there, which means that there exists at least one non-zero $\gamma_k$ with $k$ being odd, i.e. $k=2d-1$ for some $d\in \mathbb{N}$. One can notice that with this condition $\lambda_{ij} =0$ \textbf{only} for $\alpha_j = \alpha_j$. Moreover in this case it is not possible to assume any symmetry between the eigenvalues of $F$ that can lead to the situation in which the total multiplicity of some eigenvalue of $\mathbb{L}$ will be greated that the total multiplicity of zero. Therefore, the index of cyclicity in this case is the same as the \eqref{eq:atleast}, which completes the proof of the point 1.

In points 2. and 3. of the theorem 3 we consider the case when $\gamma_k$ for every $k$ being odd is equal $0$. Under this assumption the generator of evolution can be simplified and written in the form
\begin{equation}
\mathbb{L} \rho = - \sum_{k=1} ^{[\frac{N}{2}]} \gamma_{2k} [ F^{2k}, [F^{2k}, \rho]],
\end{equation}
where symbol $[x]$ denotes the integer part of $x$.

Consequently, the spectrum of the operato $\mathbb{L}$ can be written as
\begin{equation}
\sigma (\mathbb{L}) = \{ \lambda_{ij}  \in \mathbb{R}; \lambda_{ij} = - \sum_{k=1} ^{[\frac{N}{2}]} \gamma_{2k} ( \alpha_i ^{2k} - \alpha_j ^{2k})^2, \text{ for } i,j =1,\dots, r \}.
\end{equation}

Now we shall analyze two cases.

1. Case for $r=2u +1$, where $u \in \mathbb{N}$.

To get a specific result we shall assume that the eigenvalues of $F$ constitute an arithmetic sequence such that
\begin{equation}
\alpha_1 =u c \text{  where  } c\in \mathbb{R}
\end{equation}
and
\begin{equation}
\alpha_{i+1} - \alpha_i = -c \text{  for  } i=1,\dots, r-1.
\end{equation}

Bearing in mind this assumption one can notice that
\begin{equation}\label{eq:extrazero}
\alpha_i ^{2k} - \alpha_{2u+2-i}^{2k} = 0,
\end{equation}
where $i=1,2,\dots, u, u+2, u+3, \dots, 2u+1$ and $k =1,\dots, [\frac{N}{2}]$. Equality \eqref{eq:extrazero} implies that $\lambda_{i(2u+2-i)} = 0$, which means that the total multiplicity of zero as the eigenvalues of $\mathbb{L}$ can be expressed as
\begin{equation}
\eta = \sum_{i=1}^{2u+1} n_i ^2 + 2 \sum_{i=1} ^u n_i n_{2u+2-i},
\end{equation}
which is also the index of cyclicity of the generator of evolution and therefore the proof is completed.

2. Case for $r = 2u$, where $u \in \mathbb{N}$.

To get a specific result we shall assume that the eigenvalues of $F$ constitute an arithmetic sequence such that
\begin{equation}
\alpha_ i = \begin{cases}uc + (1-i) c \text{  for  } i=1, \dots, u;  \\ uc - ic \text{  for  } i = u+1, \dots, 2u. \end{cases}
\end{equation}

One can easily notice that if the spectrum of $F$ is defined in such a way, the following equality holds for $k \in \{ 1, \dots,  [\frac{N}{2}]\}$
\begin{equation}
\alpha_i ^{2k} - \alpha_{2u+1-i}^{2k} = 0,
\end{equation}
where $i=1, \dots, 2u$. This equailty implies that $\lambda_{i(2u+1-i)} = 0$. Therefore the total multiplicity of zero as the eigenvalue of $\mathbb{L}$ is equal
\begin{equation}
\eta = \sum_{i=1}^{2u} n_i ^2 + 2 \sum_{i=1}^u n_i n_{2u+1-i},
\end{equation}
which completes the proof.
\end{proof}

In this section it has been proved that if one makes specific assumptions about the character of evolution of an N-level quantum system, one can obtain concrete results concerning conditions for observability.

\section{Conclusion}

In this paper we have formulated two evolution models of $N$-level open quantum systems and we have applied to them the stroboscopic approach to quantum tomography. By making specific assumptions it was possible to determine for the analyzed generators of evolution the index of cyclicity, which physically refers to the minimal number of distinct observables that are necessary to perform quantum tomography. Because it focuses on determining the optimal criteria for quantum tomography, the stroboscopic approach has potential applications in experiments. This article contains two results connected with the criteria for quantum tomography, however the number of problems that can be solved with the stroboscopic approach is unlimited. One can treat the content of this article as a demonstration of possible applications of the stroboscopic tomography to selected evolution models. Consequently, the methods presented here can be used by one to solve more specific problems that one encounters in their research.

\section*{Acknowledgement}
This research has been supported by grant No. DEC-2011/02/A/ST1/00208 of National Science Center of Poland.

\end{document}